\documentclass[conference,twocolumn]{IEEEtran}

\def\tran{^{\mbox{\scriptsize T}}}

\usepackage{dsfont}
\usepackage{subcaption}
\usepackage{moreverb}
\usepackage{epsfig}
\usepackage{amsmath,amssymb,amsthm,mathrsfs,amsfonts,dsfont}
\usepackage{adjustbox,lipsum}
\usepackage{algorithm,algorithmic}
\usepackage{amsfonts}
\usepackage{epsfig}
\usepackage{amssymb}
\usepackage{amsmath}
\usepackage{amsthm}
\usepackage{multirow}
\usepackage{rotating}
\usepackage{graphicx}
\usepackage{tabularx}
\usepackage{array}
\usepackage{anyfontsize}
\usepackage{color,soul}
\usepackage{graphicx,dblfloatfix}
\usepackage{epstopdf}
\usepackage{blindtext}
\usepackage{amsmath}
\usepackage{amsthm,amssymb,amsmath,bm}
\usepackage{amsfonts}
\usepackage{epsfig}
\usepackage{amssymb}
\usepackage{amsmath}
\usepackage{cite}
\usepackage{graphicx}
\usepackage{fancyhdr}
\usepackage[font={small}]{caption}
\usepackage{tabularx}
\usepackage{tcolorbox}
\usepackage{cite}
\usepackage{algorithm}
\usepackage{algorithmic}
\usepackage{soul}
\usepackage{setspace}
\usepackage{enumitem}

\usepackage{arydshln}

\DeclareMathOperator*{\argmin}{arg\,min}
\newtheorem{theorem}{Theorem}
\newtheorem{lemm}{Lemma}
\newtheorem{Definition}{Definition}
\newtheorem{remark}{Remark}
\newtheorem{Pro}{Proposition}
\newtheorem{corollary}{Corollary}
\def\blue{\textcolor{blue}}

\def\yellow{\textcolor{yellow}}
\def\gray{\textcolor{gray}}
\newcommand\blfootnote[1]{%
  \begingroup
  \renewcommand\thefootnote{}\footnote{#1}%
  \addtocounter{footnote}{-1}%
  \endgroup
}

\allowdisplaybreaks

\graphicspath{{Figure/}}

\usepackage[left=0.624in ,right=0.668in ,top=0.72in,bottom=1.016in]{geometry}

\usepackage[hidelinks]{hyperref}

\begin{document}
\sloppy

\title{Status Updating with an Energy Harvesting Sensor under Partial Battery Knowledge}

\author{\IEEEauthorblockN{Mohammad Hatami\IEEEauthorrefmark{1}, Markus Leinonen\IEEEauthorrefmark{1}, and Marian Codreanu\IEEEauthorrefmark{2}}
}



\maketitle

\begin{abstract}
We consider status updating under inexact knowledge of the battery level of an energy harvesting (EH) sensor that sends status updates about a random process to users via a cache-enabled edge node. More precisely, the control decisions are performed by relying only on the battery level knowledge captured from the last received status update packet.
Upon receiving on-demand requests for fresh information from the users, the edge node uses the available information to decide whether to command the sensor to send a status update or to retrieve the most recently received measurement from the cache. We seek for the best actions of the edge node to minimize the average AoI of the served measurements, i.e., \textit{average on-demand AoI}. 
{Accounting for the partial battery knowledge}, we model the problem as a partially observable Markov decision process (POMDP), {and, through characterizing its key structures,} develop a dynamic programming algorithm to obtain an optimal policy.
Simulation results illustrate the threshold-based structure of an optimal policy and show the gains obtained by the proposed optimal POMDP-based policy compared to 
a request-aware greedy (myopic) policy.

\end{abstract}


\section{Introduction}
\blfootnote{\IEEEauthorrefmark{1}Centre for Wireless Communications, University of Oulu, Finland. \\\IEEEauthorrefmark{2}Department of Science and Technology, Link\"{o}ping University, Sweden.\\
An extended version of the work including all the proofs can be found in \cite{hatami2022POMDP_extended}.}

{In future Internet of things (IoT) systems in 5G and 6G wireless generations, timely delivery of status updates about a remotely monitored random process to a destination is the key enabler for the emerging time-critical applications, e.g., drone control and smart home systems. Such destination-centric information freshness is quantified by the \textit{age of information} (AoI) \cite{AoI_Orginal_12,sun2019age}.} {IoT networks with low-power sensors are subject to stringent energy limitations, which is often counteracted by \textit{energy harvesting} (EH) technology. Thus, there is a need for designing \textit{AoI-aware status updating} procedures that provide the end users with timely status of remotely observed processes while account for the limited energy resources of EH sensors.}

We consider a status update system consisting of {an EH sensor, an edge node, and users}. The users {are interested in time-sensitive information about a random process measured by the sensor}. The users send requests to the edge node  which has a cache storage to store the most recently received measurements from the sensor. To serve a user's request, the edge node either commands the sensor to send a fresh measurement, i.e.,  a status update packet, or uses the aged data in the cache. 
{In contrast to the existing works (e.g., \cite{Hatami-etal-20,hatami2020aoi,hatami2021spawc,hatami2022JointTcom}),} we consider a practical scenario where the edge node is informed of the sensor's battery level only via received status update packets, leading to \textit{partial} battery knowledge. Particularly, our objective is to find the best actions of the edge node to minimize the average AoI of the served measurements, i.e., \textit{average on-demand AoI}. Accounting for the {partial} battery knowledge, we model this as a partially observable Markov decision process (POMDP) problem. We convert the POMDP into a belief-state MDP and, {via characterizing its key structures}, develop {an iterative algorithm to obtain} an optimal policy. Numerical experiments illustrate the threshold-based structure of an optimal policy and show the superiority of the proposed optimal POMDP-based policy compared to a request-aware greedy policy. 
{Only a few works have applied POMDP formulation in AoI-aware design \cite{yao2020age_ISIT,gong2020AoI-Random-Arrival,shao2020partially,stamatakis2022semantics}.} In \cite{yao2020age_ISIT}, the authors proposed POMDP-based AoI-optimal {transmission scheduling in a status update system under an average energy constraint and uncertain channel state information.}
{In \cite{gong2020AoI-Random-Arrival}, the authors proposed a POMDP-based AoI-optimal scheduling policy for a multiuser uplink system under partial knowledge of the status update arrivals at the monitor node.}
{In \cite{shao2020partially},  the authors investigated AoI-optimal scheduling in a wireless sensor network where the AoI values related to the sensors are not directly observable to the access point.}
In \cite{stamatakis2022semantics}, the authors derived an efficient policy for sensor probing in an IoT network with intermittent faults and inexact knowledge about the status (healthy or faulty) of the system.

To the best of our knowledge, this is the first work that proposes an optimal policy for (on-demand) AoI minimization with an EH sensor, where the decision-making relies 
only on partial battery knowledge about the sensor's battery level. 
\section{System Model and Problem Formulation}\label{sec_systemmodel}


\subsection{Network Model}\label{sec_network}
We consider a status update system, where an energy harvesting (EH) sensor sends status updates about the monitored random process to users via a cache-enabled {gateway}, as depicted in Fig.~\ref{fig_systemmodel}. This models, e.g., an IoT sensing network, where the {gateway} represents an edge node; we refer to the gateway as the edge node henceforth. A time-slotted system with slots ${t \in \mathbb{N}}$ is considered.  


{We consider request-based status updating, where, at the beginning of slot $t$, users request for the status of the sensor (i.e., a new measurement) from the edge node. The edge node, which has a cache that stores the most recently received \textit{status update} from the sensor,  handles the arriving requests during the same slot $t$ by the following procedure.} {Let $r(t) \in \{0,1\}$, $t=1,2,\dots$ denote the random process of requesting the status of the sensor at slot $t$; $r(t) = 1$ if the status is requested (by at least one user)  and $r(t)=0$ otherwise. The requests are independent across time slots and the probability of having a request at each time slot is ${p=\mathrm{Pr}\{r(t) = 1\}}$. Upon receiving a request at slot $t$, the edge node serves the requesting user(s) by 1) commanding the sensor to send a fresh status update or 2) using the stored measurement from the cache. Let $a(t) \in \mathcal{A} = \{0,1\}$ be the \textit{command action} of the edge node at slot $t$; $a(t)=1$ if the edge node commands the sensor to send an update and $a(t)=0$ otherwise.}


\subsection{Energy Harvesting Sensor}\label{EH_model}
The sensor relies on the energy harvested from the environment. 
We model the energy arrivals $e(t) \in \left\lbrace 0 ,1\right\rbrace $, $t = 1,2,\dots$, as a Bernoulli process with rate $\lambda=\Pr\{ e(t) = 1 \}$, $\forall t$.
The sensor stores the harvested energy into a battery of finite size $B$ (units of energy). Let $b(t)$ denote the battery level of the sensor at the beginning of slot $t$, where ${b(t) \in \{0,\ldots,B\}}$.

We assume that transmitting a status update from the sensor to the edge node consumes one unit of energy (see, e.g., \cite{hatami2020aoi}).
Once the sensor is commanded (i.e., $a(t)=1$), the sensor sends an update if its battery is non-empty (i.e., $b(t) \geq 1$).
Let $d(t) \in \left\lbrace 0 ,1\right\rbrace$ denote the \textit{action of the sensor} at slot $t$; $d(t)=1$ if the sensor  sends a status update and $d(t)=0$ otherwise. Thus, $d(t) =  a(t) \mathds{1}_{\{b(t) \geq 1\}}$, where $\mathds{1}_{\{\cdot\}}$ is the indicator function. 

The evolution of the battery level is given by
\begin{equation}\label{battery_evo}
b(t+1) = \min\left\lbrace  b(t)+e(t)-d(t) , B \right\rbrace.
\end{equation}

\begin{figure}[t!]
\centering
{\includegraphics[width=0.7\columnwidth]{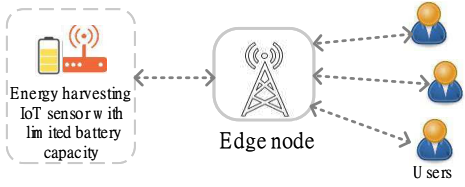}}\vspace{-2mm}
\caption{A status update system, where an EH sensor sends status updates about the monitored random process to users via a cache-enabled edge node (i.e., the gateway).}
\label{fig_systemmodel}\vspace{-5mm}
\end{figure}

\subsection{Status Updating with Partial Battery Knowledge}\label{sec_system_model_partial_battery}
We consider that the edge node is informed about the sensor's battery level (only) via the received \textit{status update packets}. {Considering this realistic setting is in stark contrast to the previous works on AoI-aware design which all assume perfect battery knowledge available at each time slot.}

Each status update packet consists of the measured value of the sensor, a  time  stamp  representing  the time when the sample was generated, and the current battery level of the sensor. Consequently, the edge node has only \textit{partial} knowledge about the battery level at each time slot, {i.e.,  \textit{outdated} knowledge based on the sensor’s last update.} 
Formally, let {$\tilde{b}(t) \in \{1,2,\dots,B\}$} denote the \emph{knowledge} about the battery level of the sensor at the edge node at slot $t$.
At slot $t$, let $u(t)$ denote the most recent slot in which the edge node received a status update packet, i.e., $u(t) = \max \{t'| t'<t, d(t') = 1 \}$.  Thus, $\tilde{b}(t) = b(u(t))$.

\subsection{On-demand Age of Information}\label{sec_AoI}
{To account for the request-based status updating, we use the notion of age of information (AoI) \cite{AoI_Orginal_12} and measure the freshness of information seen by the users via \textit{on-demand AoI} \cite{hatami2020aoi,hatami2021spawc}.}
Let $\Delta(t)$ be the AoI about the monitored random process at the edge node at the beginning of slot $t$, i.e., the number of slots elapsed since the generation of the most recently  received status update {at} the sensor. Thus, the AoI is defined as $\Delta(t) = t - u(t)$.
We make a common assumption  
that $\Delta(t)$ is upper-bounded by a finite value $\Delta^{\mathrm{max}}$, i.e.,  $\Delta(t) \in \{1, 2,\ldots ,\Delta^{\mathrm{max}}\}$. 
Besides tractability, this accounts for the fact that once the available measurement becomes excessively stale, further counting would be irrelevant.
The evolution of  $\Delta(t)$ is expressed as
\begin{equation}\label{eq_AoI}
\Delta(t+1)=
\begin{cases}
1,&\text{if} ~d(t)=1, \\
\min \{\Delta(t)+1,\Delta^{\mathrm{max}}\},&\text{if}~d(t)=0.
\end{cases}
\end{equation}
We define on-demand AoI at slot $t$ as 
\begin{equation}\label{on-demand-AoI}
\begin{array}{ll}
\Delta^\mathrm{OD}(t)  & \hspace{-2mm}\triangleq r(t) \Delta(t+1)\\
&\hspace{-2mm} = r(t) \min \{ (1-d(t)) \Delta(t)+1,\Delta^{\mathrm{max}}\}.
\end{array}
\end{equation}
In \eqref{on-demand-AoI}, since the requests come at the beginning of slot $t$ and the edge node sends measurements to the users at the end of the same slot, $\Delta(t+1)$ is the AoI seen by the users.

\subsection{Problem Formulation}\label{sec_cost}
We aim to find  the best action of the edge node at each time slot, i.e., $a(t)$, $t = 1,2,\ldots$, called an \textit{optimal policy}, that minimizes the average cost (i.e., on-demand AoI), defined as
\begin{equation}\label{eq_average_cost_persensor}
\bar{C}=\lim_{T\rightarrow\infty} \frac{1}{T}\textstyle\sum_{t=1}^{T} \mathbb{E} [\Delta^\mathrm{OD}(t)],
\end{equation}
{where the expectation is taken over all system dynamics.}

\section{POMDP Modeling} 

We model the problem of finding an optimal policy as a 
POMDP and propose an iterative algorithm that finds such an optimal policy. 
The POMDP is defined by a tuple $(\mathcal{S},\mathcal{O},\mathcal{A},\Pr(s(t+1)| s(t), a(t)),\Pr(o(t)| s(t), a(t-1)),c(s(t),a(t)))$ \cite[Chap.~7]{sigaud2013markov}, with the following elements.

\noindent $\bullet$ \textbf{State Space} $\mathcal{S}$:
Let $s(t) \in \mathcal{S}$ denote the system state at slot $t$, which is deﬁned as $s(t) = (b(t),r(t), \Delta(t),\tilde{b}(t))$.
The state space $\mathcal{S}$ has a finite dimension $|\mathcal{S}| = 2B(B+1)\Delta^{\mathrm{max}}$. We denote the \textit{observable} (visible by the edge node) part of the state by $s^{\mathrm{v}}(t) = (r(t), \Delta(t),\tilde{b}(t))$, thus $s(t) = (b(t),s^{\mathrm{v}}(t))$.

\noindent $\bullet$ \textbf{Observation Space} $\mathcal{O}$:
Let $o(t) \in \mathcal{O}$ be the edge node's observation about the system state at slot $t$, which is defined as the visible part of the state; $o(t) = s^\mathrm{v}(t)$. The observation space $\mathcal{O}$ has a finite dimension $ |\mathcal{O}| = 2B\Delta^\mathrm{max}$.

\noindent $\bullet$ \textbf{Action Space} $\mathcal{A}$:
At each time slot, the edge node decides whether to command the sensor
or not, i.e., $a(t) \in \mathcal{A} = \{0,1\}$.

\noindent $\bullet$ \textbf{State Transition Function} $\Pr(s(t+1)| s(t), a(t))$:
The state transition {function} maps a state-action pair at slot $t$ onto a distribution of states at  slot $t + 1$; the probability of transition from current state {($s(t)$) $s =(b, r, \Delta, \tilde{b})$} to next state ($s(t+1)$) $s^\prime = (b^\prime ,r^\prime,  \Delta^\prime,\tilde{b}^\prime)$  under action $a(t) = a$ is given by
\begin{equation}\label{eq_stp}
\begin{array}{ll}
& \Pr(b^\prime,r^\prime,  \Delta^\prime,\tilde{b}^\prime |b, r,  \Delta,\tilde{b} , a ) =\\ 
& \text{Pr} \big(r^\prime\big)\text{Pr} (b^\prime,\mid b , a) \text{Pr}(\Delta^\prime,\tilde{b}^\prime \mid b,\tilde{b}, \Delta,a),
\end{array}
\end{equation}
where
\begin{equation}\notag
\begin{array}{ll}
&\text{Pr} (r^\prime) = \left\lbrace 
    \begin{array}{ll}
    p, & r^\prime = 1,\\
    1-p, & r^\prime = 0,\\
    0, & \mbox{otherwise.}
    \end{array}
\right.
\end{array}
\end{equation}
\begin{equation}\notag
\begin{array}{ll}
&\text{Pr} (b^\prime\mid b = B , a = 0) = \mathds{1}_{\{b^\prime=B\}},\\
&\text{Pr} (b^\prime\mid  b < B , a = 0) = \left\lbrace 
    \begin{array}{ll}
    {\lambda,} & b^\prime = b + 1,\\
    {1-\lambda,} & b^\prime = b ,\\
    0, & \mbox{otherwise.}
    \end{array}
\right.\\
&\text{Pr} (b^\prime \mid b = 0 , a = 1) = \left\lbrace 
    \begin{array}{ll}
    {\lambda,} & b^\prime = 1,\\
    {1-\lambda,} & b^\prime = 0, \\
    0, & \mbox{otherwise.}
    \end{array}
\right.\\
&\text{Pr} (b^\prime\mid  b\geq1 , a = 1) = \left\lbrace
    \begin{array}{ll}
    {\lambda,} & b^\prime = b,\\
    {1-\lambda,} & b^\prime = b - 1,\\
    0, & \mbox{otherwise.}
    \end{array}
\right.
\end{array}
\end{equation}
\begin{equation}\notag
\hspace{-5mm}
\begin{array}{ll}
&\text{Pr} (\Delta^\prime, \tilde{b}^\prime\mid  b,\tilde{b}, \Delta, a = 0) = \mathds{1}_{\{\Delta^\prime= \min\{\Delta+1,\Delta^{\mathrm{max}}\}, \tilde{b}^\prime = \tilde{b}\}},\\
&\text{Pr} (\Delta^\prime, \tilde{b}^\prime\mid  b = 0 , \tilde{b}, \Delta, a = 1 ) = \mathds{1}_{\{\Delta^\prime= \min\{\Delta+1,\Delta^{\mathrm{max}}\}, \tilde{b}^\prime = \tilde{b}\}},\\
&\text{Pr} (\Delta^\prime, \tilde{b}^\prime\mid  b \geq 1 , \tilde{b}, \Delta, a = 1 ) = \mathds{1}_{\{\Delta^\prime= 1, \tilde{b}^\prime = b\}}.
\end{array}
\notag
\end{equation}
\noindent $\bullet$ \textbf{Observation Function}: 
The observation function is given by $\text{Pr} (o(t)| s(t),a(t-1)) = \mathds{1}_{\{ o(t) = s^{\mathrm{v}}(t)\}}$.

\noindent $\bullet$ \textbf{Immediate Cost Function} $c(s,a)$:
This is the expected one-step cost of taking action $a$ in state $s = (b,r,\Delta,\tilde{b})$,  which is calculated using \eqref{on-demand-AoI}, i.e., $c(s,a) = r[(1-a\mathds{1}_{\{b\geq1\}})\Delta+1]$.

\subsection{Belief-State 
}
In the POMDP formulation above, the system state $s(t)$
is not fully observable for the edge node -- the decision maker -- at slot $t$. In particular, the state consists of two parts as $s(t) = \{b(t),s^{\mathrm{v}}(t)\}$. Accordingly, at slot $t$, the exact battery $b(t)$ is unknown
to the edge node, whereas the request, AoI, and partial battery knowledge {captured by $s^{\mathrm{v}}(t)$} are observable.
{This incomplete state information in a POMDP causes challenges in the decision making, because the edge node can make decisions only based on the available/observed information or on the quantities derived from that information.} 

To counteract the insufficiency in the state information, we need to define state-like quantities that preserve the Markov property and summarize all the necessary information for the edge node  -- called \textit{sufficient information states}
--
in respect to searching for an optimal policy \cite[Chapter~7]{sigaud2013markov}.
{One sufficient information state is a \textit{belief-state.} We
define the belief-state at slot $t$ as $z(t) = \{\boldsymbol{\beta}(t), s^{\mathrm{v}}(t)\} \in {\mathcal{Z}(t)}$, where $\boldsymbol{\beta}(t)$ is \textit{belief}\footnote{{Note that, in general, the belief of a POMDP is a probability distribution over the entire state space $\mathcal{S}$. However, because $s^{\mathrm{v}}(t)$ is fully observable in our problem, it has no uncertainty to be modelled via a belief.}} and $\mathcal{Z}(t)$ is the belief-state space. The belief at slot $t$ is a $({B+1})$-dimensional vector ${\boldsymbol{\beta}(t) = (\beta_{0}(t),\dots,\beta_{B}(t))\tran{\in \mathcal{B}}}$, representing the probability distribution on the possible values of battery levels, where $\mathcal{B} \subset \mathbb{R}^{(B+1) \times 1}$ is the belief space. 
{Let $\phi^\mathrm{c}(t)$ be the so-called \textit{complete information state} at slot $t$, which consists of an initial probability distribution over the states, and the complete history of observations and actions up to slot $t$, i.e., $\{o(1),\dots,o(t),a(1),\dots,a(t-1)\}$.} Formally, the belief $\boldsymbol{\beta}(t)$ represents the conditional probability that the battery level is at a specific level given the complete information state $\phi^\mathrm{c}(t)$; thus, the entries of $\boldsymbol{\beta}(t)$ are defined as} 
\begin{equation}\label{eq_def_update}
    \beta_{j}(t) = \Pr(b(t) = j \mid \phi^\mathrm{c}(t)),{~j \in \{0,1,\dots,B\}.}
\end{equation} 


The belief is updated at 
each time slot based on the previous belief, the current observation, and the previous action, i.e., $\boldsymbol{\beta}(t+1) = \tau(\boldsymbol{\beta}(t),o(t+1),a(t))$, where $\tau(\cdot)$ is the belief update function, {given by the following proposition}.

\begin{Pro}\label{lemma_beleief_update}
The belief update function $\tau(\cdot)$ is given by
\begin{align}\label{eq_belief_update}
&\boldsymbol{\beta}(t+1) =  \tau(\boldsymbol{\beta}(t),o(t+1),a(t)) = \notag \\ 
&\hspace{-2mm} \left\lbrace 
\begin{array}{ll}
\hspace{-2.5mm}\boldsymbol{\Lambda}\boldsymbol{\beta}(t),\left\lbrace 
\begin{array}{ll}\hspace{-2mm} a(t) = 0,\\ \hspace{-2mm}o(t+1) = \{r(t+1),\min\{\Delta(t)+1,\Delta^{\mathrm{max}}\},\tilde{b}(t)\},
\end{array}
\right.\\ \boldsymbol{\rho}^{0},\hspace{2.8mm}\left\lbrace 
\begin{array}{ll}
\hspace{-2mm} a(t) = 1,\\ \hspace{-2mm} o(t+1) = \{r(t+1),\min\{\Delta(t)+1,\Delta^{\mathrm{max}}\},\tilde{b}(t)\},
\end{array}
\right.\\
\boldsymbol{\rho}^{1},\hspace{2.8mm}\left\lbrace 
\begin{array}{ll}\hspace{-2mm} a(t) = 1,\\\hspace{-2mm} o(t+1) = \{r(t+1),1,1\},
\end{array}
\right.\\
\cdots\\
\boldsymbol{\rho}^{B},\hspace{2.8mm}\left\lbrace 
\begin{array}{ll}\hspace{-2mm} a(t) = 1,\\\hspace{-2mm} o(t+1) = \{r(t+1),1,B\},
\end{array}
\right.
\end{array}
\right.
\end{align}
where 
the matrix $\boldsymbol{\Lambda} \in \mathbb{R}^{(B+1)\times(B+1)}$ {is a left stochastic matrix, having a banded form as}
\begin{equation}\label{eq_matrix_lambda}
\boldsymbol{\Lambda} = 
\begin{pmatrix}
1-\lambda & 0 & \cdots & 0 & 0 \\
\lambda & 1-\lambda & \cdots & 0 & 0\\
\vdots  & \vdots  & \ddots & \vdots & \vdots  \\
0 & 0 & \cdots & 1-\lambda & 0\\
0 & 0 & \cdots & \lambda & 1
\end{pmatrix},
\end{equation}
and the vectors $\boldsymbol{\rho}^{0},\boldsymbol{\rho}^{1},\boldsymbol{\rho}^{2},\ldots,\boldsymbol{\rho}^{B}$ are given by
\begin{equation}\label{eq_vectors_rho}
\begin{array}{ll}
    & {\boldsymbol{\rho}^{0}} = \boldsymbol{\rho}^{1} = \big(
    1-\lambda,\lambda,0,0,\ldots,0,0\big) \tran \\
    &\boldsymbol{\rho}^{2} = 
    \big(0,1-\lambda,\lambda,0,\ldots,0,0\big) \tran \\
    & \ldots\\
     &\boldsymbol{\rho}^{B} = 
    \big(0,0,0,0,\ldots,1-\lambda,\lambda\big) \tran.
\end{array}
\end{equation}
\end{Pro}
\begin{proof}
{The proof follows from \eqref{eq_def_update} and the Bayes' theorem. The details is presented in Appendix~\ref{sec_appendix_lemma_beleief_update}.}
\end{proof}

\section{Optimal Policy and Proposed Algorithm}
In this section, we find an optimal policy for the POMDP formulation.
Formally, a policy $\pi$ is a mapping from the belief-state space to the action space, i.e., $\pi: \mathcal{Z} \rightarrow \mathcal{A}$; the action taken in a belief-state $z(t) = z$ under policy $\pi$ is denoted by $\pi(z)$.
Under a policy $\pi$, the average cost is given by (see \eqref{eq_average_cost_persensor})
\begin{equation}\label{eq_average_cost_policy}
\bar{C}_{\pi}=\lim_{T\rightarrow\infty} \frac{1}{T}\textstyle\sum_{t=1}^{T} \mathbb{E}_{\pi}[\Delta^{\mathrm{OD}}(t) \mid z(0)],
\end{equation}
where $\mathbb{E}_{\pi}[\cdot]$ denotes the expected value of $c(t)$ given that the edge node follows $\pi$, and $z(0)$ is the initial belief-state.
{In accordance with} Section~\ref{sec_cost}, our objective is to find an optimal policy $\pi^*$ that minimizes \eqref{eq_average_cost_policy}, i.e., $\pi^* = \argmin_{\pi} \bar{C}_{\pi}$.

\begin{table}[t]
\centering
\caption{The belief space $\mathcal{B}$. {The \textit{row} and \textit{column} numbers are used to represent each belief, e.g., $\boldsymbol{\Lambda}^2\boldsymbol{\beta}$ is represented by $(0,2)$.}}
\label{tab:belief_set}
\scalebox{1}{
\begin{adjustbox}{width=1\columnwidth,center}
\begin{tabular}{c:c  c  c  c c c c}
&\gray{0}&\gray{1}&\gray{2}&\gray{3}&$\dots$&$\gray{M}$&\dots \\
\hdashline
\noalign{\vskip 0.15cm}
\gray{0}&$\boldsymbol{\beta}(1)$ & $\boldsymbol{\Lambda}\boldsymbol{\beta}(1)$ & $\boldsymbol{\Lambda}^2\boldsymbol{\beta}(1)$ & $\boldsymbol{\Lambda}^3\boldsymbol{\beta}(1)$ & $\dots$ & $\boldsymbol{\Lambda}^{M}\boldsymbol{\beta}(1)$ & $\dots$ \\
\gray{1}&$\boldsymbol{\rho}^{1}$ & $\boldsymbol{\Lambda}\boldsymbol{\rho}^{1}$ & $\boldsymbol{\Lambda}^{2}\boldsymbol{\rho}^{1}$ & $\boldsymbol{\Lambda}^{3}\boldsymbol{\rho}^{1}$ & $\dots$ & $\boldsymbol{\Lambda}^{M}\boldsymbol{\rho}^{1}$ & $\dots$\\
$\cdot$&$\cdot$ & $\cdot$ & $\cdot$ & $\cdot$ & $\cdot$ & $\cdot$ &  $\dots$ \\
$\gray{B}$&$\boldsymbol{\rho}^{B}$ & $\boldsymbol{\Lambda}\boldsymbol{\rho}^{B}$ & $\boldsymbol{\Lambda}^{2}\boldsymbol{\rho}^{B}$ & $\boldsymbol{\Lambda}^{3}\boldsymbol{\rho}^{B}$ & $\dots$ & $\boldsymbol{\Lambda}^{M}\boldsymbol{\rho}^{B}$ & $\dots$\\
&\multicolumn{6}{c}{\upbracefill}\\
&\multicolumn{6}{c}{$\scriptstyle \hat{\mathcal{B}}$}\\
\end{tabular}
\end{adjustbox}
}
\vspace{-6mm}
\end{table}

\begin{theorem}\label{theorem_bellman_eq}
An optimal policy $\pi^*$ is obtained by solving the
following equations
\begin{equation}\label{eq_ballman_pomdp}
\bar{C}^* + h(z) = \textstyle\min_{a \in \mathcal{A}} Q(z,a), z \in \mathcal{Z},
\end{equation}
where $h(z)$ is a relative value function, $\bar{C}^*$ is the optimal average cost achieved by $\pi^*$ which is independent of the initial state $z(0)$, and $Q(z, a)$ is an action-value function, which, for actions $a=0$ and $a=1$, is given by  
\begin{subequations}\label{eq_bellman_q_pomdp}
\begin{align}
    & Q(z , 0)  = r \min\{\Delta+1,\Delta^{\mathrm{max}}\} + \textstyle\sum_{r^\prime = 0}^{1} [r^\prime p + \notag \\
    &(1-r^\prime)(1-p)] h(\boldsymbol{\Lambda} \boldsymbol{\beta},r^\prime,\min\{\Delta+1,\Delta^{\mathrm{max}}\}, \tilde{b}), \\
    & Q(z ,1)  = [r \beta_{0}  \min\{\Delta+1,\Delta^{\mathrm{max}}\} + r(1-\beta_{0})] + \beta_0 \textstyle\sum_{r^\prime = 0}^{1}\notag \\
    & [r^\prime p + (1-r^\prime)(1-p)] h(\boldsymbol{\rho}^{0},r^\prime,\min\{\Delta+1,\Delta^{\mathrm{max}}\},\tilde{b})  + \notag \\
    & \textstyle\sum_{j = 1}^{B} \beta_j \big[p h(\boldsymbol{\rho}^{j},1,1,j) + (1-p) h(\boldsymbol{\rho}^{j},0,1,j) \big].
\end{align}
\end{subequations}
Then, an optimal action taken in belief-state $z$ is obtained as 
\begin{equation}\label{eq_optimal_policy}
 \pi^*(z) = \textstyle\argmin_{a\in\mathcal{A}} Q(z,a),~ z\in \mathcal{Z}.  
\end{equation}
\end{theorem}
\begin{proof}
{The proof is presented in Appendix~\ref{sec_appendix_theorem_bellman_eq}.}
\end{proof}
An optimal policy $\pi^*$ can be found by turning the Bellman's optimality equation \eqref{eq_bellman_q_pomdp} into an iterative procedure, called relative value iteration algorithm (RVIA) \cite[Section~8.5.5]{puterman2014markov}. Particularly, at each iteration $i = 0,1,\ldots$, we have
\begin{equation}\label{eq_v_itr}
\begin{array}{ll}
&V^{(i+1)}(z) =\min_{a\in\mathcal{A}} Q^{(i+1)}(z,a), \\& h^{(i+1)}(z) = V^{(i+1)}(z) - V^{(i+1)}(z_{\mathrm{ref}}),
\end{array}
\end{equation}
where ${z_{\mathrm{ref}} \in \mathcal{S}}$ is an arbitrary reference state and $Q^{(i+1)}(z,a)$, $a \in \{0,1\}$, is given by
\begin{equation}\label{eq_q_itr}
\begin{array}{ll}
    &\hspace{-4mm} Q^{(i+1)}(z , 0)  = r \min\{\Delta+1,\Delta^{\mathrm{max}}\} + \textstyle\sum_{r^\prime = 0}^{1} [r^\prime p +  \\
    &(1-r^\prime)(1-p)] h^{(i)}(\boldsymbol{\Lambda} \boldsymbol{\beta},r^\prime,\min\{\Delta+1,\Delta^{\mathrm{max}}\}, \tilde{b}), \\
    & \hspace{-4mm} Q^{(i+1)}(z ,1) \!=\! [r \beta_{0}  \!\min\{\!\Delta+1,\!\Delta^{\mathrm{max}}\}\!+\! r(1\!-\!\beta_{0})] + \beta_0\textstyle\sum_{r^\prime = 0}^{1}  \\
    & [r^\prime p + (1-r^\prime)(1-p)] h^{(i)}(\boldsymbol{\rho}^{0},r^\prime,\min\{\Delta+1,\Delta^{\mathrm{max}}\},\tilde{b})   \\
    & + \textstyle\sum_{j = 1}^{B} \beta_j \big[p h^{(i)}(\boldsymbol{\rho}^{j},1,1,j) + (1-p) h^{(i)}(\boldsymbol{\rho}^{j},0,1,j) \big].
\end{array}
\end{equation}
For any initialization $V^{(0)}(z)$, the sequences $\{Q^{(i)}(z,a)\}_{i=1,2,\ldots}$, $\{h^{(i)}(z)\}_{i=1,2,\ldots}$ and $\{V^{(i)}(z)\}_{i=1,2,\ldots}$ converge, i.e., $\lim_{i\to\infty}Q^{(i)}(z,a) = Q(z,a)$, $\lim_{i\to\infty}h^{(i)}(z) = h(z)$ and ${\lim_{i\to\infty}V^{(i)}(z) = V(z)}$, $\forall z$.
Thus, $h(z) = V(z) - V(z_{\mathrm{ref}})$ satisfies \eqref{eq_ballman_pomdp} 
and $\bar{C}^* = V(z_{\mathrm{ref}})$. 


\begin{theorem}\label{theorem_structure_v_part1}
$V(z)$ is fixed with respect to $\tilde{b}$.
\end{theorem}
\begin{proof}
{The proof is presented in Appendix~\ref{sec_appendix_theorem_structure_v_part1}.}
\end{proof}


\begin{corollary}
\normalfont {According to} Theorem~\ref{theorem_structure_v_part1}, $V(z)$, $z = \{\boldsymbol{\beta},r, \Delta, \tilde{b}\}$, {and consequently $h(z)$ and $Q(z,a)$}, do not depend on $\tilde{b}$. Thus, $\tilde{b}$ does not have any impact on calculating $\pi^*$ in \eqref{eq_optimal_policy}. Therefore, we remove\footnote{Note that while $\tilde{b}$ is removed from the belief-state, it is still needed to calculate the belief $\boldsymbol{\beta}(t)$.} $\tilde{b}$ from the belief-state $z$ and redefine the belief-state hereinafter as $z = \{\boldsymbol{\beta},r, \Delta\} \in \mathcal{Z}$. This will be exploited to reduce the computational complexity of the proposed algorithm. Note that \eqref{eq_bellman_q_pomdp} can easily be rewritten based on the new belief-state definition.
\end{corollary}

\begin{algorithm}[t]
\begin{small}
\caption{Proposed algorithm that obtains $\pi^*$.
}\label{value_itr_pomdp}
\begin{algorithmic}[1]
\STATE \textbf{Initialize} $V(z) = h(z) =0,~\forall z = \{\boldsymbol{\beta},r,\Delta\}, \boldsymbol{\beta}\in \hat{\mathcal{B}}, r \in  \{0,1\}, \Delta \in \{1,\dots,\Delta^{\mathrm{max}}\}$, determine an arbitrary $z_{\textrm{ref}} \in \mathcal{Z}$, a small threshold $\theta > 0$, and large $M$ such that $\boldsymbol{\Lambda}^M \boldsymbol{\beta} \approx \boldsymbol{\Lambda}^{(M+1)} \boldsymbol{\beta}$
\REPEAT
\FOR{$z$}
\STATE calculate $Q(z,0)$ and $Q(z,1)$ by using \eqref{eq_q_itr}
\STATE $V_{\mathrm{tmp}}(z) \leftarrow \min_{a \in \mathcal{A}} Q(z,a)$
\ENDFOR
\STATE \hspace{-2.5mm}$\delta \leftarrow \max_{z} (V_{\textrm{tmp}}(z) - V(z)) - \min_{z} (V_{\textrm{tmp}}(z) - V(z))$
\STATE \hspace{-2.5mm}$V(z) \leftarrow V_{\textrm{tmp}}(z)$ and $h(z) \leftarrow V(z) - V(z_{\mathrm{ref}})$,  for all $z$
\UNTIL{$\delta < \theta$}
\STATE $\pi^*(z) = \argmin_{a \in \mathcal{A}} Q(z,a)$, for all $z$
\end{algorithmic}
\end{small}
\end{algorithm}

Although the sequences in \eqref{eq_v_itr} converges, finding such $V(z)$ (and $h(z)$) iteratively via \eqref{eq_v_itr} is intractable, because the belief space $\mathcal{B}$ has infinite dimension. Fortunately, there is a certain pattern in the evolution of the belief-states $\{z(t)\}$, which allows us to perform a truncation of belief space $\mathcal{B}$ and subsequently develop a practical iterative algorithm relying on \eqref{eq_v_itr}.
Particularly, given the initial belief $\boldsymbol{\beta}(1)$, the belief at slot ${t = 2}$
{is chosen from the set $\{\boldsymbol{\Lambda} \boldsymbol{\beta}(1),\{\boldsymbol{\rho}^{j}\}_{j=1}^{B}\}$, the belief at ${t = 3}$ is a member of the set $\{\boldsymbol{\Lambda}^2 \boldsymbol{\beta}(1),\{\boldsymbol{\Lambda}\boldsymbol{\rho}^{j}\}_{j=1}^{B},\{\boldsymbol{\rho}^{j}\}_{j=1}^{B}\}$, the belief at ${t = 4}$ is selected from the set $\{\boldsymbol{\Lambda}^3 \boldsymbol{\beta}(1),\{\boldsymbol{\Lambda}^{2}\boldsymbol{\rho}^{j}\}_{j=1}^{B},\{\boldsymbol{\Lambda}\boldsymbol{\rho}^{j}\}_{j=1}^{B},\{\boldsymbol{\rho}^{j}\}_{j=1}^{B}\}$, and etc.} Thus, the belief space $\mathcal{B}$, which contains all the possible beliefs $\boldsymbol{\beta}(t),~\forall t$, is countable but infinite,
as shown in Table \ref{tab:belief_set}.

The following lemma expresses an important property of the matrix $\boldsymbol{\Lambda}$ {in \eqref{eq_matrix_lambda}}, which is used to truncate the belief space $\mathcal{B}$ into a finite space $\hat{\mathcal{B}}$. 
\begin{lemm}\label{lemm_struct_LAMBDA}
The $m$th power of matrix $\boldsymbol{\Lambda}$ is given by
\begin{align}
 \Lambda_{j,l}^m  =  \left \lbrace
 \begin{array}{ll}
    0,   & j < l, \\
    (1-\lambda)^m,  & j = l,\\
    \lambda^{(j-l)}(1-\lambda)^{(m-j+l)}\prod_{v = 0}^{j-l-1}\frac{(m-v)}{(v+1)},  & l < j \leq B,\\
    1 - \sum_{j^\prime = 1}^{B}\Lambda_{j^\prime,l}^m,  & j = B + 1,
 \end{array}
 \right.
 \notag
\end{align}
where $\Lambda_{j,l}^m$ is the {entry} of matrix $\boldsymbol{\Lambda}^m$ at its $j$th row and $l$th column.
Thus, for any $\lambda > 0$, we have
\begin{equation}
    \lim_{m \rightarrow \infty} \Lambda_{j,l}^m = \left \lbrace
    \begin{array}{ll}
    0,   & j \leq B, \\
    1,  & j = B + 1,
    \end{array}
    \right.
\end{equation}
\end{lemm}
\begin{proof}
{The proof is presented in Appendix~\ref{sec_appendix_lemm_struct_LAMBDA}.}
\end{proof}
\begin{corollary}
\normalfont By Lemma~\ref{lemm_struct_LAMBDA}, for a sufficiently large integer $M$, we have $\boldsymbol{\Lambda}^{M} \approx \boldsymbol{\Lambda}^{M+1}$. Thus, we construct a truncated belief space $\hat{\mathcal{B}}$ of finite dimension $|\hat{\mathcal{B}}| = (B+1)M$, as shown in Table \ref{tab:belief_set}. Intuitively, the value $M$ represents {the maximum number of consecutive no-command actions ($a = 0$) for which the belief is updated; from the $(M+1)$th no-command onward, the belief is no longer updated}.
{This is reasonable because after $M$ consecutive $a = 0$ actions, the battery is almost full,
i.e., $\boldsymbol{\Lambda}^M \boldsymbol{\beta} \rightarrow (0,0,\dots,0,1)\tran$,~$\forall \boldsymbol{\beta} \in \mathcal{B}$, and thus, for sufficiently large $M$, the space $\hat{\mathcal{B}}$ contains (almost) all the possible beliefs.}
\end{corollary}

Finally, considering the truncated belief space $\hat{\mathcal{B}}$, we use \eqref{eq_v_itr}--\eqref{eq_q_itr} to find $V(z)$, $h(z)$, $Q(z,a)$, and consequently an optimal policy $\pi^*$ iteratively, as presented in Algorithm~\ref{value_itr_pomdp}. 

\section{Simulation Results}\label{sec_simulation}
We consider a scenario with $\lambda = 0.08$, $p = 0.8$, $\Delta^{\mathrm{max}} = 64$, and $B = 2$.
{Fig.~\ref{fig_structure} illustrates the structure of an optimal policy $\pi^*$, where each point represents a potential  belief-state as a three-tuple $z = (\boldsymbol{\beta},1,{\Delta})$. For each such $z$, a blue point indicates that the optimal action is to command the sensor (i.e., $\pi^*(z) = 1$), whereas a red point means not to command.}
Fig.~\ref{fig_structure} illustrates that $\pi^*$ has a \textit{threshold-based} structure with respect to the AoI.
To exemplify, consider the belief-state $z =\{ (1,7), 1, 20\}$
in which $\pi^*(z) = 1$; then, by the threshold-based structure,  $\pi^*(\underline{z}) = 1$ for all $\underline{z} = \{ (1,10), 1, \Delta\}$, $\Delta \geq 20$.

Fig.~\ref{fig_performance} depicts the performance of the proposed algorithm. In the \textit{ request-aware greedy (myopic)} policy, the edge node commands the sensor whenever there is a request (i.e., $r(t) = 1$).
{As benchmark, we 
consider a case that the edge node knows the exact battery level at each time slot. In this case, an optimal policy, denoted by $\pi_{\mathrm{MDP}}(s), \forall s$, can be found by the value iteration algorithm \cite{hatami2021spawc}. Clearly, this policy
serves as a lower bound to the proposed POMDP-based algorithm.
As shown in Fig.~\ref{fig_performance}, for sufficiently large $M$ (say, $M \geq 32$), the proposed algorithm obtains optimal performance.}
Furthermore, the proposed method reduces the average cost by
{$28~\%$} compared to the greedy policy.



\newcommand\fw{0.98}
\begin{figure}
    \centering
    \includegraphics[width=\fw \columnwidth]{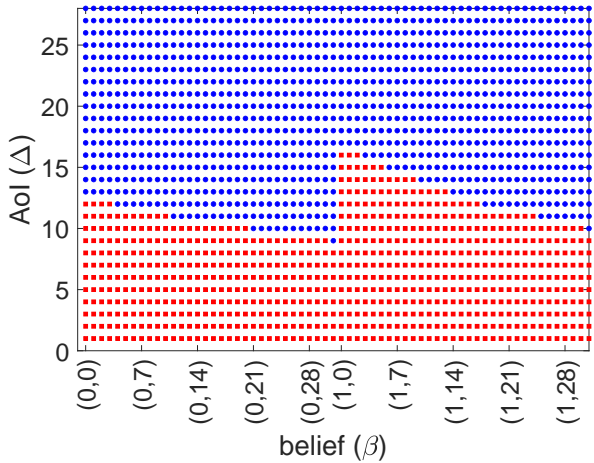}
    \caption{Structure of an optimal policy $\pi^*(z)$ for each belief-state $z = \{\boldsymbol{\beta},1, {\Delta}\}$, where $p = 0.8$, $\lambda = 0.08$, and initial belief $\boldsymbol{\beta}(1) = (1/3,1/3,1/3)$. \mbox{Red: no command $a=0$; blue: command $a=1$.}}
    \label{fig_structure}
\end{figure}


\newcommand\fww{0.93}
\begin{figure}
    \centering
    \includegraphics[width=\fww \columnwidth]{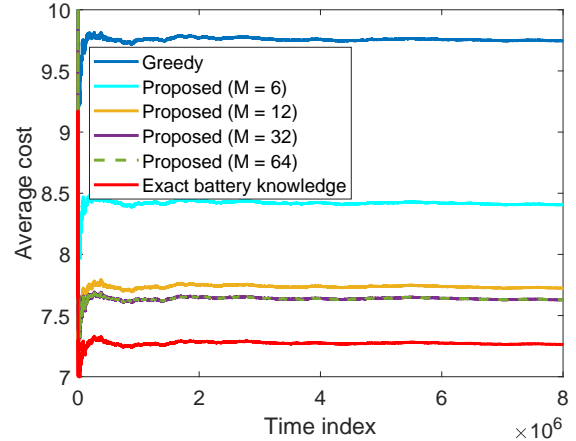}
    \vspace{-1mm}
    \caption{Performance of the proposed algorithm.}
    \label{fig_performance}
    \vspace{-5mm}
\end{figure}



\section{Conclusions}\label{sec_conclusions}
{We considered  status updating under inexact knowledge of the battery level of an EH sensor that sends updates to users via a cache-enabled edge node. Accounting for the partial battery knowledge at the edge node, we derived a POMDP model for the on-demand AoI minimization problem. We converted the POMDP into a belief-state MDP and, via characterizing its key structures, developed an iterative algorithm that obtains an optimal policy. We numerically depicted that an optimal POMDP-based policy has a threshold-based structure and demonstrated the performance gains obtained by the proposed algorithm compared to a request-aware greedy policy.}

\section{Acknowledgments}
This research has been financially supported by the Infotech Oulu, the Academy of Finland (grant 323698), and Academy of Finland 6Genesis Flagship (grant 318927). M. Hatami would like to acknowledge the support of Nokia Foundation. The work of M. Leinonen has also been financially supported by the Academy of Finland (grant 340171).



\bibliographystyle{IEEEtran}
\bibliography{Bib/conf_short,Bib/IEEEabrv,Bib/Bibliography}

\newpage
.
\newpage

\begin{onecolumn}
\begin{appendix}

\subsection{Proof of Proposition \ref{lemma_beleief_update}}\label{sec_appendix_lemma_beleief_update}
We start from the definition of the belief in \eqref{eq_def_update} and express ${\beta_j(t+1)}$ as
\begin{equation}\label{eq-proof-belief-update-def}
    \begin{aligned}
     \beta_j(t+1) & = \Pr\left(b(t+1) = j \mid \phi^\mathrm{c}(t+1)\right) \\
     & = \Pr\left(b(t+1) = j \mid \phi^\mathrm{c}(t), o(t+1), a(t)\right)\\ 
    & = \frac{\Pr(b(t+1) = j, \phi^\mathrm{c}(t), o(t+1), a(t))}{\Pr(\phi^\mathrm{c}(t), o(t+1), a(t))} \\
    & = \frac{\Pr(\phi^\mathrm{c}(t),a(t)) \Pr(b(t+1) = j,o(t+1)\mid \phi^\mathrm{c}(t),a(t))}{\Pr(\phi^\mathrm{c}(t),a(t)) \Pr(o(t+1)\mid \phi^\mathrm{c}(t),a(t))}\\
    & = \frac{\Pr(b(t+1) = j,o(t+1)\mid \phi^\mathrm{c}(t),a(t))}{\Pr(o(t+1)\mid \phi^\mathrm{c}(t),a(t))}\\
    & \overset{(a)}{=} \frac{\Pr(b(t+1) = j,o(t+1)\mid \phi^\mathrm{c}(t),a(t))}{\zeta}\\
    & {=} \frac{1}{\zeta}\sum_{i = 0}^{B}\Pr(b(t) = i,b(t+1) = j,o(t+1)\mid \phi^\mathrm{c}(t),a(t)) \\
    & = \frac{1}{\zeta}\sum_{i = 0}^{B}\Pr(b(t) = i \mid \phi^\mathrm{c}(t),a(t))\Pr(b(t+1) = j \mid b(t) = i, \phi^\mathrm{c}(t),a(t)) \times \\& \hspace{13mm} \Pr(o(t+1)\mid b(t+1) = j, b(t) = i, \phi^\mathrm{c}(t),a(t))\\
    & \overset{(b)}{=} \frac{1}{\zeta}\sum_{i = 0}^{B} \beta_i(t) \Pr(b(t+1) = j \mid b(t) = i,a(t)) \Pr(o(t+1)\mid b(t+1) = j, b(t) = i, \phi^\mathrm{c}(t),a(t))
    \end{aligned}
\end{equation}
where $(a)$ follows by introducing a normalization factor ${\zeta \triangleq \Pr(o(t+1)\mid \phi^\mathrm{c}(t),a(t))}$, which can be calculated using the fact that ${\sum_j \beta_j(t+1) = 1}$, and $(b)$ follows from i) ${\Pr(b(t) = i \mid \phi^\mathrm{c}(t),a(t))=\Pr(b(t) = i \mid \phi^\mathrm{c}(t))}$ because $b(t)$ is given when performing action $a(t)$, and subsequently using the belief definition $\beta_i(t)$ in \eqref{eq_def_update}, ii) $\Pr(b(t+1) = j \mid b(t) = i,\phi^\mathrm{c}(t),a(t)) = \Pr(b(t+1) = j \mid b(t) = i,a(t))$ because $b(t+1)$ is independent of $\phi^\mathrm{c}(t)$ given $b(t)$ and $a(t)$.

Next, we derive an expression for $\beta_j(t+1)$ in \eqref{eq-proof-belief-update-def} for the different cases regarding action ${a(t) \in \{0,1\}}$.

\subsubsection{Action $a(t) = 0$}
For the case where $a(t) = 0$, the edge node does not receive a status update packet and thus the next observation $o(t+1)$ is either $o(t+1) = \{1,\min\{\Delta(t)+1,\Delta^{\mathrm{max}}\},\tilde{b}(t)\}$ or $o(t+1) = \{0,\min\{\Delta(t)+1,\Delta^{\mathrm{max}}\},\tilde{b}(t)\}$, which happens with probability $p$ and $1-p$, respectively. Recall that $p$ is the probability of having a request at each time slot (i.e., $\mathrm{Pr}\{r(t) = 1\} = p$, $\forall t$). We next calculate the belief update function for the case where $a(t) = 0$ and $o(t+1) = \{1,\min\{\Delta(t)+1,\Delta^{\mathrm{max}}\},\tilde{b}(t)\}$. By \eqref{eq-proof-belief-update-def}, ${\beta_j(t+1)}$, $j \in \{0,1,\dots,B\}$, is expressed as
\begin{equation}\label{eq-proof-belief-update-def_a0}
    \begin{aligned}
    \beta_j(t+1) &=  \frac{1}{\zeta}\sum_{i = 0}^{B} \beta_i(t) \Pr(b(t+1) = j \mid b(t) = i,a(t) = 0) \underbrace{\Pr(o(t+1) \mid b(t+1) = j, b(t) = i, \phi^\mathrm{c}(t),a(t)=0)}_{\overset{(a)}{=} p} \\ 
    & = \frac{p}{\zeta}\sum_{i = 0}^{B} \beta_i(t) \underbrace{\Pr(b(t+1) = j \mid b(t) = i,a(t) = 0)}_{(\star)},
\end{aligned}
\end{equation}
where $(a)$ follows from 
\begin{equation}\label{eq-proof-belief-update-prob-o1-a0}
    \begin{aligned}
     & \Pr\Big(o(t+1) =  \{1,\min\{\Delta(t)+1,\Delta^{\mathrm{max}}\},\tilde{b}(t)\} \mid   b(t+1)= j, b(t) = i, \underbrace{\{\phi^\mathrm{c}(t-1),r(t),\Delta(t),\tilde{b}(t),a(t-1)\}}_{\phi^\mathrm{c}(t)},a(t)=0\Big)=\\
     & \Pr\Big(o(t+1) =  \{1,\min\{\Delta(t)+1,\Delta^{\mathrm{max}}\},\tilde{b}(t)\} \mid   \Delta(t), \tilde{b}(t),a(t)=0\Big)=\\
     &  \Pr\Big(r(t+1) = 1,\Delta(t+1) = \min\{\Delta(t)+1,\Delta^{\mathrm{max}}\}, \tilde{b}(t+1) = \tilde{b}(t) \mid \Delta(t), \tilde{b}(t),a(t)=0\Big) \overset{(b)}{=} \\
     & \underbrace{\Pr(r(t+1) = 1)}_{=p} \underbrace{\Pr(\Delta(t+1) = \min\{\Delta(t)+1,\Delta^{\mathrm{max}}\},\tilde{b}(t+1) = \tilde{b}(t)  \mid \Delta(t),\tilde{b}(t), a(t)=0)}_{=1} = p,
    \end{aligned}
\end{equation}
where $(b)$ follows from the independence of the request process from the other variables.
Recall that at each slot, the sensor harvests one unit of energy with probability $\lambda$. Thus, $(\star)$ in \eqref{eq-proof-belief-update-def_a0} is expressed as
\begin{equation} \label{eq-proof-belief-update-Pr-harvest}
    \begin{aligned}
    & \Pr(b(t+1) = j \mid b(t) = i<B,a(t) = 0) = \left\{ 
    \begin{array}{ll}
    1-\lambda, & j = i, \\
    \lambda, & j = i+1,\\
    0, & \mbox{otherwise.}
    \end{array}
    \right.\\
    & \Pr(b(t+1) = j \mid b(t) = B,a(t) = 0) = \left\{ 
    \begin{array}{ll}
    1, & j = B, \\
    0, & \mbox{otherwise.}
    \end{array}
    \right.
    \end{aligned}
\end{equation}
By substituting \eqref{eq-proof-belief-update-Pr-harvest} into \eqref{eq-proof-belief-update-def_a0}, we can express $\beta_j(t+1)$, for each $j \in \{0,1,\ldots,B\}$, as
\begin{equation}\label{eq-proof-belief-update-def_a0_summarised}
    \begin{aligned}
    &\beta_0(t+1) = \frac{p}{\zeta}(1-\lambda) \beta_0(t),\\
    &\beta_1(t+1) = \frac{p}{\zeta}(\lambda \beta_0(t)+(1-\lambda)\beta_1(t)),\\
    &\beta_2(t+1) = \frac{p}{\zeta}(\lambda \beta_1(t)+(1-\lambda)\beta_2(t)),\\
    & \dots,\\
    & \beta_{B-1}(t+1) = \frac{p}{\zeta}(\lambda \beta_{B-2}(t)+(1-\lambda)\beta_{B-1}(t)),\\
    & \beta_{B}(t+1) = \frac{p}{\zeta}(\lambda \beta_{B-1}(t)+\beta_{B}(t)).
    \end{aligned}
\end{equation}
Using $\sum_{j = 0}^{B}\beta_j(t+1) = 1$, we can easily calculate the normalization factor to be ${\zeta = p}$.
By rewriting \eqref{eq-proof-belief-update-def_a0_summarised} in the vector form, the updated belief is written as
$\boldsymbol{\beta}(t+1) = \boldsymbol{\Lambda} \boldsymbol{\beta}(t)$, where the matrix $\boldsymbol{\Lambda}$ is defined in \eqref{eq_matrix_lambda}. 
For the case where $a(t) = 0$ and ${o(t+1) = \{0,\min\{\Delta(t)+1,\Delta^{\mathrm{max}}\},\tilde{b}(t)\}}$, one can follow the similar steps and conclude that ${\boldsymbol{\beta}(t+1) = \boldsymbol{\Lambda} \boldsymbol{\beta}(t)}$ as well.

\subsubsection{Action $a(t) = 1$} For the case where ${a(t) = 1}$, the edge node receives a status update packet whenever $b(t) \geq 1$ and does not receive an update whenever $b(t) = 0$. In this regard, if ${b(t)  = m \geq 1}$, the next observation is either $o(t+1) = \{1,1,m\}$ or $o(t+1) = \{0,1,m\}$, ${m \in \{1,2,\dots,B\}}$; and, if $b(t) = 0$, the next observation is either $o(t+1) = \{1,\min\{\Delta(t)+1,\Delta^{\mathrm{max}}\},\tilde{b}(t)\}$ or $o(t+1) = \{0,\min\{\Delta(t)+1,\Delta^{\mathrm{max}}\},\tilde{b}(t)\}$.
We next calculate the belief update function for these cases. Starting with the case where $a(t) = 1$ and ${o(t+1) = \{1,\min\{\Delta(t)+1,\Delta^{\mathrm{max}}\},\tilde{b}(t)\}}$, by \eqref{eq-proof-belief-update-def}, we have
\begin{equation}\label{eq-proof-belief-update-def_a1_o10}
    \begin{aligned}
    \beta_j(t+1)  & =  \frac{1}{\zeta}\sum_{i = 0}^{B} \beta_i(t) \Pr(b(t+1) = j \mid b(t) = i,a(t) = 1) \Pr(o(t+1)\mid b(t+1) = j, b(t) = i, \phi^\mathrm{c}(t),a(t) = 1) \\
    & = \frac{1}{\zeta} \beta_0(t) \Pr(b(t+1) = j \mid b(t) = 0,a(t) = 1) \underbrace{\Pr(o(t+1)\mid b(t+1) = j, b(t) = 0, \phi^\mathrm{c}(t),a(t) = 1)}_{\overset{(a)}{=} p }  \\
    & + \frac{1}{\zeta} \sum_{i = 1}^{B} \beta_i(t) \Pr(b(t+1) = j \mid b(t) = i,a(t) = 1) \underbrace{\Pr(o(t+1)\mid b(t+1) = j, b(t) = i, \phi^\mathrm{c}(t),a(t) = 1)}_{\overset{(b)}{=} 0} \\ & = \frac{p\beta_0(t)}{\zeta}  \Pr(b(t+1) = j \mid b(t) = 0,a(t) = 1) \\ 
    & = \left\{ 
    \begin{array}{ll}
    \frac{p\beta_0(t)}{\zeta} (1-\lambda),&  j = 0,\\
    \frac{p\beta_0(t)}{\zeta} \lambda,&  j = 1,\\
    0,& j = 2,\ldots,B,
    \end{array}
    \right.
\end{aligned}
\end{equation}
where $(a)$ follows similarly as \eqref{eq-proof-belief-update-prob-o1-a0} and  $(b)$ follows from
\begin{equation}\label{eq-proof-belief-update-prob-o11-a1}
    \begin{aligned}
     & \Pr\Big(o(t+1) =  \{1,\min\{\Delta(t)+1,\Delta^{\mathrm{max}}\},\tilde{b}(t)\} \mid   b(t+1)= j, b(t) = i \geq 1, \phi^\mathrm{c}(t) ,a(t)=1\Big)=\\
     & \Pr(r(t+1) = 1) \underbrace{\Pr\big(\Delta(t+1) = \min\{\Delta(t)+1,\Delta^{\mathrm{max}}\},\tilde{b}(t+1) = \tilde{b}(t)  \mid b(t+1)= j, b(t) = i \geq 1, \phi^\mathrm{c}(t) ,a(t)=1\big)}_{\overset{(c)}{=}0} = 0,
    \end{aligned}\notag
\end{equation}
where $(c)$ follows because the edge node receives a status update packet whenever $a(t) = 1$ and $b(t) \geq 1$, and thus, $\Delta(t+1) = 1$ (see \eqref{eq_AoI}).
Using $\sum_{j = 0}^{B}\beta_j(t+1) = 1$, the normalization factor can readily be derived to be ${\zeta = p\beta_0(t)}$.
By rewriting \eqref{eq-proof-belief-update-def_a1_o10} in the vector form, we conclude that $\boldsymbol{\beta}(t+1) = \boldsymbol{\rho}^0$, where the vector $\boldsymbol{\rho}^0$ is defined in \eqref{eq_vectors_rho}. For the case where $a(t) = 1$ and ${o(t+1) = \{0,\min\{\Delta(t)+1,\Delta^{\mathrm{max}}\},\tilde{b}(t)\}}$, one can follow the similar steps and conclude that $\boldsymbol{\beta}(t+1) = \boldsymbol{\rho}^0$ as well.

For the cases where $a(t) = 1$ and ${o(t+1) = \{1,1,m\}}$, $m \in  \{1,2,\dots,B\}$, by \eqref{eq-proof-belief-update-def}, we have
\begin{equation}\label{eq-proof-belief-update-def_a1_m_summarised}
    \begin{aligned}
    & \beta_j(t+1)  =  \frac{1}{\zeta}\sum_{i = 0}^{B} \beta_i(t) \Pr(b(t+1) = j \mid b(t) = i,a(t) = 1) \underbrace{\Pr(o(t+1)\mid b(t+1) = j, b(t) = i, \phi^\mathrm{c}(t),a(t) = 1)}_{\overset{(a)}{=} p \mathds{1}_{\{i=m\}}} \\ & = \frac{p\beta_m(t)}{\zeta}  \Pr(b(t+1) = j \mid b(t) = m,a(t) = 1) \\ 
    & = \left\{ 
    \begin{array}{ll}
    \frac{p\beta_m(t)}{\zeta} (1-\lambda),&  j = m-1,\\
    \frac{p\beta_m(t)}{\zeta} \lambda,&  j = m,\\
    0,& \mbox{otherwise.}
    \end{array}
    \right.
\end{aligned}
\end{equation}
where $\mathds{1}_{\{\cdot\}}$ is the indicator function and $(a)$ follows from
\begin{equation}
    \begin{aligned}
     & \Pr\Big(o(t+1) =  \{1,1,m\} \mid   b(t+1)= j, b(t) = i, \phi^\mathrm{c}(t) ,a(t)=1\Big)=\\
     & \underbrace{\Pr(r(t+1) = 1)}_{= p} \underbrace{\Pr\big(\Delta(t+1) = 1,\tilde{b}(t+1) = m \geq 1  \mid b(t+1)= j, b(t) = i, \phi^\mathrm{c}(t) ,a(t)=1\big)}_{\overset{(b)}{=} \mathds{1}_{\{i=m\}}} = p \mathds{1}_{\{i=m\}}.
    \end{aligned}\notag
\end{equation}
where $(b)$ follows because the edge node receives a status update packet whenever $a(t) = 1$ and $b(t) \geq 1$, and thus, $\tilde{b}(t+1) = b(t)$ (see Section~\ref{sec_system_model_partial_battery}).
Using $\sum_{j = 0}^{B}\beta_j(t+1) = 1$, the normalization factor is derived as $\zeta = p\beta_m(t)$.
Therefore, by \eqref{eq-proof-belief-update-def_a1_m_summarised}, we have $\boldsymbol{\beta}(t+1) = \boldsymbol{\rho}^m$, where the vectors $\boldsymbol{\rho}^m$, $m \in \{1,2,\dots,B\}$, are defined in \eqref{eq_vectors_rho}. For the cases where $a(t) = 1$ and ${o(t+1) = \{0,1,m\}}$, $m \in  \{1,2,\dots,B\}$, one can follow the similar steps and conclude that $\boldsymbol{\beta}(t+1) = \boldsymbol{\rho}^m$.

\subsection{Proof of Theorem \ref{theorem_bellman_eq}}\label{sec_appendix_theorem_bellman_eq}

By rewriting the Bellman equation for the average cost POMDP \cite[Chapter~7.1]{sigaud2013markov}, \cite[Chapter~7]{krishnamurthy2016partially}, we have
\begin{equation}\notag
    \bar{C}^* + h(z) = \min_{a \in \mathcal{A}}[c(z,a) + \sum_{o^\prime}\Pr(o^\prime \mid z,a)h(z^\prime)],~ z \in \mathcal{Z},
\end{equation}
where $c(z,a)$ is the expected immediate cost obtained by choosing action $a$ in belief-state $z$, ${z = (\boldsymbol{\beta},o) = (\boldsymbol{\beta},r,\Delta,\tilde{b})}$ is the current belief state, ${o^\prime = (r^\prime,\Delta^\prime,\tilde{b}^\prime)}$ is the observation given action $a$, and ${z^\prime = (\tau(\boldsymbol{\beta},o^\prime,a),r^\prime,\Delta^\prime,\tilde{b}^\prime)}$ is the next belief state given action $a$ and observation $o^\prime$.
By defining an action-value function as $Q(z,a)\triangleq c(z,a) + \sum_{o^\prime}\Pr(o^\prime \mid z,a)h(z^\prime)$, we have
\begin{equation}\label{eq-proof-bellman-Qfunc}
\begin{aligned}
    Q(z,a) & =   c(z,a) +  \sum_{o^\prime}\Pr(o^\prime \mid z,a)h(z^\prime)\\
    & = \sum_{s} \Pr(\underbrace{s}_{b,s^\mathrm{v}} \mid z) c(s,a) + \sum_{o^\prime}\sum_{s} \underbrace{\Pr(o^\prime,s \mid z,a)}_{\Pr(s\mid z,a)\Pr(o^\prime \mid s,z,a)} h(z^\prime) \\
    & = \sum_{b} \sum_{s^\mathrm{v}} \underbrace{\Pr(b,s^\mathrm{v} \mid \boldsymbol{\beta},o)}_{\mathds{1}_{\{s^\mathrm{v} = o\}}\Pr(b\mid\boldsymbol{\beta},s^\mathrm{v},o)} c(b,s^\mathrm{v},a) + \sum_{r^\prime}\sum_{\Delta^\prime}\sum_{\tilde{b}^\prime}\sum_{b} \sum_{s^\mathrm{v}} \underbrace{\Pr(b,s^\mathrm{v}\mid \boldsymbol{\beta},o,a)}_{\mathds{1}_{\{s^\mathrm{v} = o\}}\Pr(b\mid\boldsymbol{\beta},s^\mathrm{v},o,a)}\Pr(r^\prime,\Delta^\prime,\tilde{b}^\prime \mid b,s^\mathrm{v},\boldsymbol{\beta},o,a) h(z^\prime)\\
    & = \sum_{b} \beta_b c(b,o,a) + \sum_{r^\prime}\sum_{\Delta^\prime}\sum_{\tilde{b}^\prime}\sum_{b} \beta_b \Pr(r^\prime,\Delta^\prime,\tilde{b}^\prime\mid b,\underbrace{r,\Delta,\tilde{b}}_{o},a) h(z^\prime)\\
    & = \sum_{b} \beta_b c(b,o,a) + \sum_{b} \beta_b \sum_{r^\prime}\Pr(r^\prime) \sum_{\Delta^\prime} \sum_{\tilde{b}^\prime} \Pr(\Delta^\prime,\tilde{b}^\prime\mid b,\Delta,\tilde{b},a) h(z^\prime),
\end{aligned}
\end{equation}
where $s^\mathrm{v} = \{r,\Delta,\tilde{b}\}$ is the visible part of the state, which is equivalent to the observation $o$ (i.e., $o = s^\mathrm{v}$).
For the case where $a = 0$, by Proposition~\ref{lemma_beleief_update}, we have ${z^\prime = (\tau(\boldsymbol{\beta},o^\prime,0),r^\prime,\Delta^\prime,\tilde{b}^\prime) = (\boldsymbol{\Lambda}\boldsymbol{\beta},r^\prime,\min\{\Delta+1,\Delta^\mathrm{max}\},\tilde{b}^\prime)}$, and thus, $Q(z,a = 0)$ is expressed as
\begin{equation}\label{eq_bellman_action0}
\begin{aligned}
    Q(z,0)  & = Q(\boldsymbol{\beta},r,\Delta,\tilde{b},0)  = \sum_{b} \beta_b \underbrace{c(b,r,\Delta,\tilde{b},a = 0)}_{r \min\{\Delta+1,\Delta^\mathrm{max}\}} + \sum_{b} \beta_b \sum_{r^\prime}\Pr(r^\prime) \sum_{\Delta^\prime} \sum_{\tilde{b}^\prime} \underbrace{\Pr(\Delta^\prime,\tilde{b}^\prime\mid b,\Delta,\tilde{b},a = 0)}_{\mathds{1}_{\{\Delta^\prime = \min\{\Delta+1,\Delta^\mathrm{max}\},\tilde{b}^\prime = \tilde{b}\}}} h(z^\prime) \\ & =   r \min\{\Delta+1,\Delta^\mathrm{max}\} \underbrace{\sum_{b} \beta_b}_{=1} + \sum_{r^\prime}\underbrace{[r^\prime p+ (1-r^\prime)(1-p)]}_{\Pr(r^\prime)} h(\boldsymbol{\Lambda}\boldsymbol{\beta},r^\prime,\min\{\Delta+1,\Delta^\mathrm{max}\},\tilde{b}) \underbrace{\sum_{b} \beta_b}_{=1}\\ & =   r \min\{\Delta+1,\Delta^\mathrm{max}\} + \sum_{r^\prime}[r^\prime p+ (1-r^\prime)(1-p)] h(\boldsymbol{\Lambda}\boldsymbol{\beta},r^\prime,\min\{\Delta+1,\Delta^\mathrm{max}\},\tilde{b})
\end{aligned}\notag
\end{equation}
For the case where $a = 1$, $Q(z,a = 1)$ is expressed as
\begin{equation}\label{eq_bellman_action1}
\begin{aligned}
    Q(z,1)  & = Q(\boldsymbol{\beta},r,\Delta,\tilde{b},1) = \sum_{b = 0}^{B} \beta_b c(b,r,\Delta,\tilde{b},a = 1) + \sum_{b = 0}^{B} \beta_b \sum_{r^\prime}\Pr(r^\prime) \sum_{\Delta^\prime} \sum_{\tilde{b}^\prime} \Pr(\Delta^\prime,\tilde{b}^\prime\mid b,\Delta,\tilde{b},a = 1) h(z^\prime) \\ & = \beta_0 \underbrace{c(b = 0,r,\Delta,\tilde{b},a = 1)}_{r \min\{\Delta+1,\Delta^\mathrm{max}\}} + \sum_{b = 1}^{B} \beta_b \underbrace{c(b \geq 1,r,\Delta,\tilde{b},a = 1)}_{r} +\beta_0 \sum_{r^\prime}\Pr(r^\prime) \sum_{\Delta^\prime} \sum_{\tilde{b}^\prime} \underbrace{\Pr(\Delta^\prime,\tilde{b}^\prime\mid b = 0,\Delta,\tilde{b},a =1)}_{\mathds{1}_{\{\Delta^\prime = \min\{\Delta+1,\Delta^\mathrm{max}\},\tilde{b}^\prime = \tilde{b}\}}} h(z^\prime) \\ &  + \sum_{b = 1}^{B} \beta_b \sum_{r^\prime}\Pr(r^\prime) \sum_{\Delta^\prime} \sum_{\tilde{b}^\prime} \underbrace{\Pr(\Delta^\prime,\tilde{b}^\prime\mid b\geq 1,\Delta,\tilde{b},a=1)}_{\mathds{1}_{\{\Delta^\prime = 1 ,\tilde{b}^\prime = b\}}} h(z^\prime) \\ & = r \beta_0 \min\{\Delta+1,\Delta^\mathrm{max}\} + r \underbrace{\sum_{b = 1}^{B} \beta_b}_{1-\beta_0}  \\ & + \beta_0 \sum_{r\prime} [r^\prime p+ (1-r^\prime)(1-p)] h(\underbrace{\tau(\boldsymbol{\beta},r^\prime,\min\{\Delta+1,\Delta^\mathrm{max}\},\tilde{b},a = 1)}_{\overset{(a)}{=}\boldsymbol{\rho}^0},r^\prime,\min\{\Delta+1,\Delta^\mathrm{max}\},\tilde{b}) \\ & + \sum_{b = 1}^{B} \beta_b \sum_{r^\prime}[r^\prime p+ (1-r^\prime)(1-p)] h(\underbrace{\tau(\boldsymbol{\beta},r^\prime,1,b,a = 1)}_{\overset{(a)}{=}\boldsymbol{\rho}^b},r^\prime,1,b)\\ & = r \beta_0 \min\{\Delta+1,\Delta^\mathrm{max}\} + r (1-\beta_0) + \beta_0 \sum_{r^\prime = 0}^{1} [r^\prime p+ (1-r^\prime)(1-p)] h(\boldsymbol{\rho}^0,r^\prime,\min\{\Delta+1,\Delta^\mathrm{max}\},\tilde{b}) \\ & + \sum_{b = 1}^{B} \beta_b \sum_{r^\prime = 0}^{1}[r^\prime p+ (1-r^\prime)(1-p)] h(\boldsymbol{\rho}^b,r^\prime,1,b), 
\end{aligned}\notag
\end{equation}
where $(a)$ follows from Proposition~\ref{lemma_beleief_update}.

\subsection{Proof of Theorem \ref{theorem_structure_v_part1}}\label{sec_appendix_theorem_structure_v_part1}

\newcommand\Deltabar{\underline{\Delta}}
\newcommand\zbar{\underline{z}}
\newcommand\bbar{\underline{\tilde{b}}}

We consider two belief-states $z = (\boldsymbol{\beta},r,\Delta,\tilde{b})$ and $\zbar = (\boldsymbol{\beta},r,\Delta,\bbar)$, $\tilde{b} \leq \bbar$, and prove that  i) $V(z) \leq V(\zbar)$ and ii) $V(z) \geq V(\zbar)$. Therefore, we conclude that $V(z) = V(\zbar)$.

(i) Since the sequence $\{V^{(i)}(z)\}_{{i=1,2,\ldots}}$ converges to $V(z)$ for any initialization, it suffices to prove that $V^{(i)}(\zbar) \geq V^{(i)}(z)$, $\forall{i}$, which is shown using mathematical induction.
The  initial values are selected arbitrarily, e.g., $V^{(0)}(z) = 0$ and $V^{(0)}(\zbar) = 0$, hence, $V^{(i)}(z) \leq V^{(i)}(\zbar)$ holds for $i = 0$. 
Assume that
$V^{(i)}(z) \leq V^{(i)}(\zbar)$ for some $i$; we need to prove that ${V^{(i+1)}(z) \leq V^{(i+1)}(\zbar)}$.
Let us denote an optimal action in belief-state $z$ at iteration $i = 1,2,\dots$ by $\pi^{(i)}(z)$, which is  given by $\pi^{(i)}(z) = \argmin_{a}Q^{(i)}(z,a)$. 
Thus, we have 
\begin{equation}
\begin{array}{ll}
 V^{(i+1)}(z) - V^{(i+1)}(\zbar) & = \min_{a} Q^{(i+1)}(z,a) -  \min_{a} Q^{(i+1)}(\zbar,a)  \\ & =
Q^{(i+1)}(z,\pi^{(i+1)}(z)) -  Q^{(i+1)}(\zbar,\pi^{(i+1)}(\zbar)) \\ & \overset{(a)}{\leq} Q^{(i+1)}(z,\pi^{(i+1)}(\zbar)) -  Q^{(i+1)}(\zbar,\pi^{(i+1)}(\zbar)),\notag
\end{array}
\end{equation}
where $(a)$ follows from the fact that taking action $\pi^{(i+1)}(\zbar)$ in state $z$ is not necessarily optimal. We show that  ${Q^{(i+1)} (z,\pi^{(i+1)}(\zbar))-  Q^{(i+1)}(\zbar,\pi^{(i+1)}(\zbar)) \leq 0}$ for all possible actions ${\pi^{(i+1)} (\zbar)\in \{0,1\}}$. For the case where $\pi^{(i+1)}(\zbar) = 0$, we have 
\begin{equation}\label{eq_proof_Q0}
\begin{array}{ll}
    & Q^{(i+1)}(z,0) -  Q^{(i+1)}(\zbar,0)  = \\  
    &\textstyle\sum_{r^\prime = 0}^{1} [r^\prime p + (1-r^\prime)(1-p)] \underbrace{[V^{(i)}(\boldsymbol{\Lambda} \boldsymbol{\beta},r^\prime,\min\{\Delta+1,\Delta^{\mathrm{max}}\}, \tilde{b}) - V^{(i)}(\boldsymbol{\Lambda} \boldsymbol{\beta},r^\prime,\min\{\Delta+1,\Delta^{\mathrm{max}}\}, \bbar)]}_{(a)\leq 0} \leq 0,
    \end{array}
\end{equation}
where step $(a)$ follows from the induction assumption.
For the case where $\pi^{(i+1)}(\zbar) = 1$, we have
\begin{equation}\label{eq_proof_Q1}
\begin{array}{ll}
    & Q^{(i+1)}(z,1) -  Q^{(i+1)}(\zbar,1)  = \\ 
    &\beta_0 \textstyle\sum_{r^\prime = 0}^{1}[r^\prime p + (1-r^\prime)(1-p)]\underbrace{[ V^{(i)}(\boldsymbol{\rho}^{0},r^\prime,\min\{\Delta+1,\Delta^{\mathrm{max}}\},\tilde{b}) - V^{(i)}(\boldsymbol{\rho}^{0},r^\prime,\min\{\Delta+1,\Delta^{\mathrm{max}}\},\bbar)}_{(a) \leq 0} \leq 0,
    \end{array}
\end{equation}
where step $(a)$ follows from the induction assumption.

(ii) Similarly, we show that $V^{(i)}(\zbar) \leq V^{(i)}(z)$, $\forall i$, which is shown using mathematical induction.
The  initial values are selected arbitrarily, e.g., $V^{(0)}(z) = 0$ and $V^{(0)}(\zbar) = 0$, hence, $V^{(i)}(\zbar) \leq V^{(i)}(z)$ holds for $i = 0$. 
Assume that
$V^{(i)}(\zbar) \leq V^{(i)}(z)$ for some $i$; we need to prove that ${V^{(i+1)}(\zbar) \leq V^{(i+1)}(z)}$.
We have 
\begin{equation}
\begin{array}{ll}
 V^{(i+1)}(\zbar) - V^{(i+1)}(z) & = 
 \min_{a} Q^{(i+1)}(\zbar,a) -  \min_{a} Q^{(i+1)}(z,a) \\
 & = Q^{(i+1)}(\zbar,\pi^{(i+1)}(\zbar)) -  Q^{(i+1)}(z,\pi^{(i+1)}(z))  \\ &\overset{(a)}{\leq} Q^{(i+1)}(\zbar,\pi^{(i+1)}(z)) -  Q^{(i+1)}(z,\pi^{(i+1)}(z)),\notag
\end{array}
\end{equation}
where $(a)$ follows from the fact that taking action $\pi^{(i+1)}(z)$ in belief-state $\zbar$ is not necessarily optimal. We show that  ${Q^{(i+1)} (\zbar,\pi^{(i+1)}(z))-  Q^{(i+1)}(z,\pi^{(i+1)}(z)) \leq 0}$ for all possible actions ${\pi^{(i+1)} (z)\in \{0,1\}}$. For the case where $\pi^{(i+1)}(z) = 0$, we have 
\begin{equation}\label{eq_proof_Q0_case2}
\begin{array}{ll}
    & Q^{(i+1)}(\zbar,0) -  Q^{(i+1)}(z,0)  = \\  
    &\textstyle\sum_{r^\prime = 0}^{1} [r^\prime p + (1-r^\prime)(1-p)] \underbrace{[V^{(i)}(\boldsymbol{\Lambda} \boldsymbol{\beta},r^\prime,\min\{\Delta+1,\Delta^{\mathrm{max}}\}, \bbar) - V^{(i)}(\boldsymbol{\Lambda} \boldsymbol{\beta},r^\prime,\min\{\Delta+1,\Delta^{\mathrm{max}}\}, \tilde{b})]}_{(a)\leq 0} \leq 0,
    \end{array}
\end{equation}
where step $(a)$ follows from the induction assumption.
For the case where $\pi^{(i+1)}(\zbar) = 1$, we have 
\begin{equation}\label{eq_proof_Q1_case2}
\begin{array}{ll}
    & Q^{(i+1)}(\zbar,1) -  Q^{(i+1)}(z,1)  = \\ 
    &\beta_0 \textstyle\sum_{r^\prime = 0}^{1}[r^\prime p + (1-r^\prime)(1-p)]\underbrace{[ V^{(i)}(\boldsymbol{\rho}^{0},r^\prime,\min\{\Delta+1,\Delta^{\mathrm{max}}\},\bbar) - V^{(i)}(\boldsymbol{\rho}^{0},r^\prime,\min\{\Delta+1,\Delta^{\mathrm{max}}\},\tilde{b})}_{(a) \leq 0} \leq 0,
    \end{array}
\end{equation}
where in step $(a)$ follows from the induction assumption.

\subsection{Proof of Lemma \ref{lemm_struct_LAMBDA}}\label{sec_appendix_lemm_struct_LAMBDA}
We prove this lemma by mathematical induction. For $m = 1$, we have 
\begin{equation}
\boldsymbol{\Lambda} = 
\begin{pmatrix}
1-\lambda & 0 & \cdots & 0 & 0 \\
\lambda & 1-\lambda & \cdots & 0 & 0\\
\vdots  & \vdots  & \ddots & \vdots & \vdots  \\
0 & 0 & \cdots & 1-\lambda & 0\\
0 & 0 & \cdots & \lambda & 1
\end{pmatrix},
\end{equation}
and hence, the lemma holds for $m = 1$. Assume that the lemma holds for some $m$, i.e., $\boldsymbol{\Lambda}^m$ has the form as shown below
\begin{equation}\label{eq_matrix_lambda_power_m}
\boldsymbol{\Lambda}^m = 
\begin{pmatrix}
(1-\lambda)^m & 0 & \cdots & 0 & 0 \\
m\lambda(1-\lambda)^{m-1} & (1-\lambda)^m & \cdots & 0 & 0\\
\frac{m(m-1)}{2!}\lambda^2(1-\lambda)^{m-2} & m\lambda(1-\lambda)^{m-1} & \cdots & 0 & 0\\
\frac{m(m-1)(m-2)}{3!}\lambda^3(1-\lambda)^{m-3} & \frac{m(m-1)}{2}\lambda^2(1-\lambda)^{m-2} & \cdots & 0 & 0\\
\vdots  & \vdots & \ddots & \vdots & \vdots  \\
\lambda^{B-1}(1-\lambda)^{m-B+1}\prod_{\nu = 0}^{B-2} \frac{(m-\nu)}{\nu+1}& \lambda^{B-2}(1-\lambda)^{m-B+2}\prod_{\nu = 0}^{B-3} \frac{(m-\nu)}{\nu+1}&  \cdots   & (1-\lambda)^m& 0\\
1- \sum_{j^\prime}\Lambda_{j^\prime,1} & 1- \sum_{j^\prime}\Lambda_{j^\prime,2} & \cdots & 1-(1-\lambda)^m & 1
\end{pmatrix}.
\end{equation}
We prove that the lemma also holds for $m+1$. We have 
\begin{equation}\label{eq_matrix_lambda_power_m+1}
\begin{array}{ll}
& \boldsymbol{\Lambda}^{m+1} = \boldsymbol{\Lambda}^m \boldsymbol{\Lambda} \\
= &
\begin{pmatrix}
(1-\lambda)^m & 0 & \cdots & 0 & 0 \\
m\lambda(1-\lambda)^{m-1} & (1-\lambda)^m & \cdots & 0 & 0\\
\frac{m(m-1)}{2!}\lambda^2(1-\lambda)^{m-2} & m\lambda(1-\lambda)^{m-1} & \cdots & 0 & 0\\
\frac{m(m-1)(m-2)}{3!}\lambda^3(1-\lambda)^{m-3} & \frac{m(m-1)}{2}\lambda^2(1-\lambda)^{m-2} & \cdots & 0 & 0\\
\vdots  & \vdots & \ddots & \vdots & \vdots  \\
\lambda^{B-1}(1-\lambda)^{m-B+1}\prod_{\nu = 0}^{B-2} \frac{(m-\nu)}{\nu+1}& \lambda^{B-2}(1-\lambda)^{m-B+2}\prod_{\nu = 0}^{B-3} \frac{(m-\nu)}{\nu+1}&  \cdots   & (1-\lambda)^m& 0\\
1- \sum_{j^\prime}\Lambda_{j^\prime,1} & 1- \sum_{j^\prime}\Lambda_{j^\prime,2} & \cdots & 1-(1-\lambda)^m & 1
\end{pmatrix} \times \\
&\begin{pmatrix}
1-\lambda & 0 & \cdots & 0 & 0 \\
\lambda & 1-\lambda & \cdots & 0 & 0\\
\vdots  & \vdots  & \ddots & \vdots & \vdots  \\
0 & 0 & \cdots & 1-\lambda & 0\\
0 & 0 & \cdots & \lambda & 1
\end{pmatrix}  \\
=&
\begin{pmatrix}
(1-\lambda)^{m+1} & 0 & \cdots & 0 & 0 \\
(m+1)\lambda(1-\lambda)^{m} & (1-\lambda)^{m+1} & \cdots & 0 & 0\\
\frac{(m+1)m}{2!}\lambda^2(1-\lambda)^{m-1} & (m+1)\lambda(1-\lambda)^{m} & \cdots & 0 & 0\\
\frac{(m+1)m(m-1)}{3!}\lambda^3(1-\lambda)^{m-2} & \frac{(m+1)m}{2}\lambda^2(1-\lambda)^{m-1} & \cdots & 0 & 0\\
\vdots  & \vdots & \ddots & \vdots & \vdots  \\
\lambda^{B-1}(1-\lambda)^{m-B+2}\prod_{\nu = 0}^{B-2} \frac{(m+1-\nu)}{\nu+1}& \lambda^{B-2}(1-\lambda)^{m-B+3}\prod_{\nu = 0}^{B-3} \frac{(m-\nu+1)}{\nu+1}&  \cdots   & (1-\lambda)^{m+1}& 0\\
1- \sum_{j^\prime}\Lambda_{j^\prime,1} & 1- \sum_{j^\prime}\Lambda_{j^\prime,2} & \cdots & 1-(1-\lambda)^{m+1} & 1
\end{pmatrix}.
\end{array}
\end{equation}

\end{appendix}

\end{onecolumn}

\end{document}